\documentclass[11pt,reqno,a4paper,UKenglish,cleveref, autoref, thm-restate]{amsart}
\usepackage[utf8]{inputenc}
\usepackage{graphicx,multicol,latexsym,amsmath,amssymb}
\usepackage{mathabx} 
\usepackage[square,numbers]{natbib}
\numberwithin{equation}{section}

\date{30 December, 2022}

\title[Depth-First Search in a random digraph] 
{Depth-First Search performance in a random digraph with geometric 
outdegree distribution}

\author{Philippe {Jacquet}}
\address{Inria Saclay Ile de France,  France} 
\email{philippe.jacquet@inria.fr}
\author{Svante {Janson}}
\thanks{Supported by the Knut and Alice Wallenberg Foundation}
\address{Department of Mathematics, Uppsala University, PO Box 480,
SE-751~06 Uppsala, Sweden}
\email{svante.janson@math.uu.se}
\newcommand\urladdrx[1]{{\urladdr{\def~{{\tiny$\sim$}}#1}}}
\urladdrx{http://www2.math.uu.se/~svante/}


\thanks{We thank Donald Knuth for posing us questions and conjectures that led to the present paper.}

\newtheorem{theorem}{Theorem}[section]
\newtheorem{lemma}[theorem]{Lemma}

\newtheorem{corollary}[theorem]{Corollary}
\newtheorem{conjecture}[theorem]{Conjecture}
\theoremstyle{definition}

\newtheorem{remark}[theorem]{Remark}
\newenvironment{romanenumerate}[1][-10pt]{
\addtolength{\leftmargini}{#1}\begin{enumerate}
 \renewcommand{\labelenumi}{\textup{(\roman{enumi})}}%
 \renewcommand{\theenumi}{\textup{(\roman{enumi})}}%
 }{\end{enumerate}}

\newcounter{thmenumerate}

\def\td{\widetilde{d}}
\def\tl{\widetilde{\ell}}

\begingroup
  \divide\count255 by 60
  \multiply\count255 by -60
  \advance\count255 by \time
  \ifnum \count255 < 10 \xdef\klockan{\the\count1.0\the\count255}
  \else\xdef\klockan{\the\count1.\the\count255}\fi
\endgroup

\newcommand\ie{\emph{i.e.}}
\newcommand\eg{e.g.\spacefactor=1000}
\newcommand\set[1]{\ensuremath{\{#1\}}}

\newcommand\Bigset[1]{\ensuremath{\Bigl\{#1\Bigr\}}}

\newcommand\xpar[1]{(#1)}
\newcommand\bigpar[1]{\bigl(#1\bigr)}
\newcommand\Bigpar[1]{\Bigl(#1\Bigr)}

\newcommand\lrpar[1]{\left(#1\right)}
\newcommand\bigsqpar[1]{\bigl[#1\bigr]}
\newcommand\sqpar[1]{[#1]}
\newcommand\Bigsqpar[1]{\Bigl[#1\Bigr]}

\newcommand\cpar[1]{\{#1\}}

\newcommand\abs[1]{\lvert#1\rvert}
\newcommand\bigabs[1]{\bigl\lvert#1\bigr\rvert}
\newcommand\Bigabs[1]{\Bigl\lvert#1\Bigr\rvert}

\newcommand\E{\operatorname{\mathbb E{}}}
\newcommand\PP{\operatorname{\mathbb P{}}}
\newcommand\Var{\operatorname{Var}}
\newcommand\Cov{\operatorname{Cov}}
\newcommand\Ge{\operatorname{Ge}}
\newcommand\Gei{\operatorname{Ge}_1}
\newcommand\Po{\operatorname{Po}}

\newcommand\Bin{\operatorname{Bin}}

\newcommand\NegBin{\operatorname{NegBin}}

\newcommand{\tend}{\longrightarrow}
\newcommand\dto{\overset{\mathrm{d}}{\tend}}
\newcommand\pto{\overset{\mathrm{p}}{\tend}}

\newcommand\dd{\,\mathrm{d}}
\newcommand{\refT}[1]{Theorem~\ref{#1}}
\newcommand{\refTs}[1]{Theorems~\ref{#1}}
\newcommand{\refC}[1]{Corollary~\ref{#1}}

\newcommand{\refL}[1]{Lemma~\ref{#1}}

\newcommand{\refR}[1]{Remark~\ref{#1}}

\newcommand{\refS}[1]{Section~\ref{#1}}

\newcommand{\refSS}[1]{Section~\ref{#1}}

\newcommand\ga{\alpha}
\newcommand\gb{\beta}

\newcommand\gD{\Delta}
\newcommand\gf{\varphi}

\newcommand\gl{\lambda}

\newcommand\go{\omega}

\newcommand\gs{\sigma}

\newcommand\gss{\sigma^2}
\newcommand\gth{\theta}
\newcommand\gu{\upsilon}
\newcommand\gU{\Upsilon}
\newcommand\eps{\varepsilon}
\newcommand\gthx{\gth^*}

\newcommand\cF{\mathcal F}

\newcommand\cM{\mathcal M}

\newcommand\bbR{\mathbb{R}}

\newcommand\indic[1]{\boldsymbol1\cpar{#1}}

\newcommand\ceil[1]{\lceil#1\rceil}

\newcommand\floor[1]{\lfloor#1\rfloor}

\newcommand\upto{\nearrow}
\newcommand\downto{\searrow}
\newcommand\qw{^{-1}}
\newcommand\qww{^{-2}}
\newcommand\qq{^{1/2}}
\newcommand\qqw{^{-1/2}}

\newcommand{\sumtn}{\sum_{t=1}^n}

\newcommand\oi{\ensuremath{[0,1]}}

\newcommand\oio{[0,1)}

\newcommand\ntoo{\ensuremath{{n\to\infty}}}

\newcounter{case}

\newcommand\pfitem[1]{\par\smallskip\noindent #1:}

\newcommand\ER{Erd\H os--R\'enyi}

\newcommand\op{o_{\mathrm{p}}}

\newcommand\Oll{O_{L^2}}
\newcommand\Ollnqq{\Oll(n\qq)}

\newcommand\whp{w.h.p.\spacefactor=1000}
\newcommand\tlp{\tl^+}
\newcommand\bd{\overline d}
\newcommand\Tx{T^*}
\newcommand\Mx{M^*}

\newcommand\hpi{\widehat\pi}
\newcommand\bpi{\overline\pi}
\newcommand\hxi{\widehat\xi}
\newcommand\htd{\widehat{\td}}
\newcommand\hd{\widehat{d}}
\newcommand\htl{\widehat{\tl}}

\newcommand\bJ{\overline{J}}
\newcommand\JJ{J}
\newcommand\PQ{\mathbf{P}}

\newcommand\meta{\eta}

\newcommand\dff{depth-first forest}

\newcommand\setn{[n]}

\newcommand\xtau{x}

\newcommand\rhoo{\rho_0}

\newcommand\TTx{\mathbf{T}}

\newcommand\etam{\eta^<}
\newcommand\etal{\eta^=}
\newcommand\etap{\eta^>}
\newcommand\chit{\tau}
\newcommand\chix{\chi}
\newcommand\bmu{\overline{\mu}}
\newcommand\gssx{v}

\begin{document}

\begin{abstract}

We present an analysis of the depth-first search algorithm in a random
digraph model with independent outdegrees having a
geometric distribution.

The results include asymptotic results for the depth profile of vertices, the
height (maximum depth) and average depth, the number of trees in the forest,
the size 
of the largest and second-largest trees, and the numbers of arcs of different
types in the depth-first jungle. 
Most results are first order. For the height we show an asymptotic normal
distribution.

This analysis proposed by Donald Knuth in his next to appear volume of 
\emph{The Art of Computer Programming} 
gives interesting insight in one of the most elegant
and efficient algorithm for graph analysis due to Tarjan.
\end{abstract}

\maketitle

\section{Introduction}
The motivation of this paper is a new section in Donald Knuth's 
\emph{The Art of Computer  Programming} 
\cite{Knuth12A}, which is dedicated to Depth-First Search (DFS) in a
digraph.
Briefly, the DFS starts with an arbitrary vertex, and explores the arcs from
that vertex one by one. When an arc is found leading to a  vertex that
has not been seen before, the DFS explores the arcs from that vertex in the
same way,
in a recursive fashion, before returning to the next arc from its parent.
This eventually yields a tree containing all descendants of the the first
vertex (which is the root of the tree). 
If there still are some unseen vertices, the DFS starts again with one of
them and finds a new tree, and so on until all vertices are found.
We refer to \cite{Knuth12A} for details as well as
for historical notes. 
(See also \ref{stack1}--\ref{stack2} in \refS{Sgen}.)
Note that the digraphs in \cite{Knuth12A} and here are
multi-digraphs, where loops and multiple arcs are allowed. 
(Although in our random model these are few and usually not important.)
The DFS algorithm
generates a spanning forest 
(the \emph{depth-first forest})
in the digraph, with all arcs in the forest directed away from the
roots. Our main purpose is to study the 
properties of the depth-first forest,
starting with a random digraph $G$;
in particular we study
the distribution of the depth of vertices in the depth-first forest.

The random digraph model that we consider 
(following \citet{Knuth12A})
has $n$ vertices and a given
outdegree distribution $\PQ$, which in the main part of the present paper
is a geometric distribution $\Ge(1-p)$ for some fixed $0<p<1$.
The outdegrees (number of outgoing arcs) of the $n$ vertices are independent
random numbers with this distribution.
The endpoint of each arc is uniformly selected at random among the $n$
vertices, independently of all other arcs. (Therefore, an arc can loop back
to the starting vertex, and multiple arcs can occur.)
We consider asymptotics as $n\to\infty$ for a fixed outdegree distribution.

In the present paper, we study the case of a geometric outdegree distribution
in detail; we also (in \refS{SGeo1}) briefly give corresponding results for 
the shifted geometric outdegree distribution $\Gei(1-p)$, and discuss the
similarities and differences between the two cases.
The case of a general outdegree distribution (with  finite variance)
will be studied in a forthecoming paper \cite{SJ364-general}, where we 
use a somewhat different method which allows us to 
extend many (but not all) of the results in the present paper and
obtain similar, but partly weaker, results; 
see also \refS{Sgen}.
One reason for studying the geometric case separately is that its
lack-of-memory property leads to 
interesting features and simplifications not present for
general outdegree distributions; this is seen both in \cite{Knuth12A} 
and in the proofs and results below. 
In particular, 
the depth process studied in \refS{SGeo}
will be a Markov chain, which is the basis of our analysis.

In addition to studying the \dff, we give also 
(in \refSS{SSGeoTypes})
some results on the number of
edges of different types in the
\emph{depth-first jungle};
this is defined in \cite{Knuth12A}
as the original
digraph with arcs classified by 
the DFS algorithm 
into the following five types, see Figure~\ref{fig:forest} for examples:
\begin{itemize}
    \item \emph{loops};
    \item \emph{tree arcs}, the arcs in the resulting depth-first forest;
    \item \emph{back arcs}, 
the arcs which point to an ancestor of the current vertex in the current tree;
    \item \emph{forward arcs}, 
the arcs which point to an already discovered descendant of the current
vertex in the current tree;
    \item \emph{cross arcs}, 
all other arcs (these point to an already discovered vertex which is neither
a descendant nor an ancestor of the current vertex, and might be in another
tree).
\end{itemize}
(See further the exercises in \cite{Knuth12A}.)

\begin{figure}[ht]
\includegraphics[height=7cm]{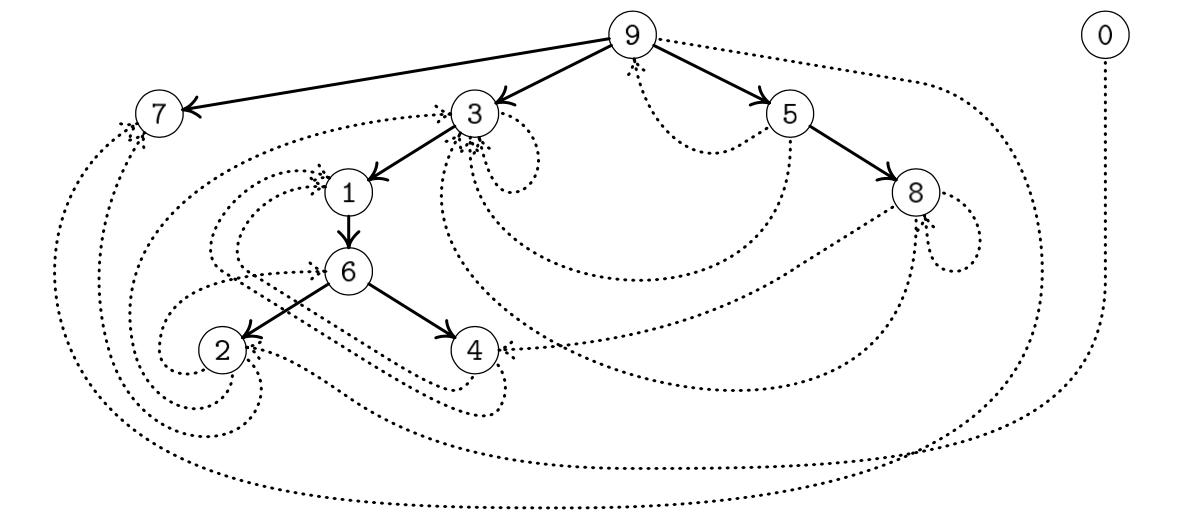}
\caption{Example of a depth-first forest (jungle) from \cite{Knuth12A}, 
by courtesy of Donald Knuth. Tree arcs are solid ({\it e.g.} \textcircled{9}$\to$ \textcircled{3}). For example, 
\textcircled{3}$\dasharrow$ \textcircled{3} is a loop,
\textcircled{2}$\dasharrow$ \textcircled{3} is a back arc, 
\textcircled{9}$\dasharrow$ \textcircled{7} is a forward arc,
\textcircled{8}$\dasharrow$ \textcircled{4} and \textcircled{0}$\dasharrow$ \textcircled{2} are cross arcs. 
}
\label{fig:forest}
\end{figure}

\begin{remark}\label{RER}
Some related results for DFS in an undirected \ER{} graph $G(n,\gl/n)$
are proved by \citet{EnriquezFM} and \citet{DiskinK}, and
DFS in a random \ER{} digraph has been studied
for example in the proof of \cite[Theorem 3]{KrivSud2013}.
These models are closely related to our model with a Poisson
outdegree distribution $\PQ$; they will therefore
be further discussed in \cite{SJ364-general}. 
\end{remark}

\begin{remark}\label{Rdep}
We consider only the case of a fixed outdegree
distribution $\PQ$. The results can be extended to distributions $\PQ_n$
depending on $n$, under suitable conditions. This is particularly
interesting in the critical case, with expectations $\gl_n\to1$;
however, this is out of the scope of the present paper.
\end{remark}

The main results for a geometric out-degree distribution
are stated and proved in \refS{SGeo}.
We analyze the process $d(t)$ of depths of the vertices, in the order they
are found by the DFS.  For a geometric outdegree distribution (but not in
general), $d(t)$ is a Markov chain, and we find its first-order
limit by a martingale argument; 
moreover, we  show Gaussian fluctuations. 
We also find results for the numbers of different types of arcs defined
above; this
includes verifying some conjectures from previous versions of
\cite{Knuth12A}.

In \refS{SGeo1} we study briefly the case of a shifted geometric outdegree
distribution. The same method as in \refS{SGeo} 
works in this case too, but the explicit
results are somewhat different. One motivation for this section is to show
that some of the relations found in \refS{SGeo} for a geometric outdegree
distribution do not hold for arbitrary distributions.

We end in \refS{Sgen}
with some comments on the case of general outdegree distributions.


\subsection{Some notation}
We denote the given outdegree distribution by $\PQ$.
Recall that our standing assumption is that the outdegrees of the vertices are 
i.i.d.\ (independent and identically distributed).

The mean outdegree, \ie, the expectation  of $\PQ$, is denoted by
$\gl$. 
In analogy with branching processes, we say that the 
random digraph 
is
\emph{subcritical} if $\gl<1$,
\emph{critical} if $\gl=1$, 
and
\emph{supercritical} if $\gl>1$.

As usual, \whp{} means \emph{with high probability}, \ie, with probability
$1-o(1)$ as \ntoo.
We use $\pto$ for convergence in probability, and $\dto$ for convergence in
distribution of random variables.

Moreover,
let $(a_n)$ be a sequence of positive numbers, and $X_n$ a sequence of
random variables. 
We write $X_n=\op(a_n)$ if, as $\ntoo$, 
$X_n/a_n\pto0$,
\ie, if for every $\eps>0$, we have
$\PP(|X_n|>\eps a_n)\to0$.
Note that this is equivalent to the existence of a sequence $\eps_n\to0$ such
that 
$\PP(|X_n|>\eps_n a_n)\to0$, or in other words $|X_n|\le\eps_na_n$ w.h.p.
(This is sometimes denoted ``$X_n=o(a_n)$ \whp'', but we will not use this 
notation.)

Furthermore, 
$X_n=\Oll(a_n)$ means $\E\bigsqpar{|X_n/a_n|^2}=O(1)$.
Note that 
$X_n=\Oll(a_n)$ implies 
$X_n=\op(\go_na_n)$,
for any sequence
$\omega_n\to\infty$.
Note also that $X_n=\Oll(a_n)$ implies $\E X_n=O(a_n)$;
thus error terms of this type implies immediately estimates for
expectations and second moments.
In particular, for the most common case below,
$X_n=\Ollnqq$ is equivalent to $\E X_n=O(n\qq)$ and 
$\Var X_n= O(n)$. 

$\Ge(1-p)$ denotes the geometric distribution on \set{0,1,\dots};
thus $\eta\sim\Ge(1-p)$ means that $\eta$ is a random variable with
$\PP(\eta=k)=p^k(1-p)$, $k\ge0$.
Similarly, $\Gei(1-p)$ denotes the shifted geometric distribution on
\set{1,2,\dots};
thus $\eta\sim\Gei(1-p)$ means
$\PP(\eta=k)=p^{k-1}(1-p)$, $k\ge1$.
$\Po(\gl)$ denotes a Poisson distribution with mean $\gl$.

We define $\rhoo(x)$, for $x\ge0$,
as the largest solution in $[0,1)$ to
\begin{align}\label{rhoo}
1-\rhoo= e^{-x\rhoo}
.\end{align}
As is well known, $\rhoo(x)$
is the survival probability of a Galton--Watson process with 
a Poisson offspring distribution $\Po(x)$ with mean $x$. 
We have $\rhoo(x)=0$ for $x\le1$ and $0<\rhoo(x)<1$ for $x>1$.
(See \eg{} \cite[Theorem I.5.1]{AN}.) 

For a real number $x$, we write $x^+:=\max\set{x,0}$.
$\setn:=\set{1,\dots,n}$.
All logarithms are natural. 
$C$ and $c$ are sometimes used for  positive constants.

\begin{remark}
We  state many results with error estimates in $L^2$, which means
estimates on the second moment; we conjecture that
the results extend to higher moment and estimates in $L^p$ for any
$p<\infty$,
but we have not pursued this.
\end{remark}


\section{Depth analysis with geometric outdegree distribution} \label{SGeo}
In this section we assume that the outdegree distribution is geometric
$\Ge(1-p)$ for some fixed $0<p<1$, and thus has mean 
\begin{align}\label{glGeo}
\gl:=\frac{p}{1-p}.
\end{align}

When doing the DFS on a random digraph of the type studied in this paper,
it is natural to reveal the outdegree of a vertex as soon as we find it.
(See \ref{stack1}--\ref{stack2} in \refS{Sgen}.)
However, for a geometric outdegree distribution, 
because of its lack-of-memory property,
we do not have to immediately reveal the outdegree when we find a new vertex
$v$.
Instead, we only check whether there is at least one
outgoing arc (probability $p$),  
and if so, we find its endpoint and explore
this endpoint if it has not already been visited; eventually, we return to
$v$, and then we check whether there is another outgoing arc (again
probability $p$, 
by the lack-of-memory property), 
and so on.
 This will yield the important Markov property in the construction in the
next subsection.

In the following, by a \emph{future arc} from some vertex, we mean an arc
that at the current time has not yet been seen by the DFS.

\subsection{Depth Markov chain}
Our aim is to track the evolution of the search depth as a function of the
number $t$ of discovered vertices. Let $v_t$ be the $t$-th vertex
discovered by  the DFS ($t=1,\dots,n$), and let $d(t)$ be the depth of $v_t$
in the resulting depth-first forest, \ie, the number of tree edges that connect the root of the current tree to  $v_t$.
The first found vertex $v_1$ is a root, and thus $d(1)=0$.

The quantity $d(t)$ follows a Markov chain with transitions ($1\le t<n$):
\begin{romanenumerate}
    \item {$d(t+1)=d(t)+1$}.
\\This happens 
if, for some $k\ge1$, $v_t$ has at least $k$ outgoing arcs, 
the first $k-1$ arcs lead to vertices already visited, and the $k$th arc leads
to a new vertex (which then becomes $v_{t+1}$). The probability of this is
\begin{align}\label{new}
  \sum_{k=1}^\infty p^k \Bigpar{\frac{t}{n}}^{k-1}\Bigpar{1-\frac{t}{n}}
=\frac{(1-t/n)p}{1-pt/n}.
\end{align}

\item{$d(t+1)=d(t)$, assuming $d(t)>0$}.
\\This holds if all arcs from $v_t$ lead to already visited vertices, \ie,
(i) does not happen, and furthermore, the parent of $v_t$ has at least one
future arc leading to an unvisited vertex.
These two events are independent. 
Moreover, by the lack-of-memory property, 
the number of future arcs from the parent of $v_t$ 
has also the distribution $\Ge(1-p)$. Hence, 
the probability that one of these future arcs leads to an unvisited vertex 
equals the probability in 
\eqref{new}. The probability of (ii) is thus
\begin{align}\label{ii}
\Bigpar{1-\frac{(1-t/n)p}{1-pt/n}}\frac{(1-t/n)p}{1-pt/n}.
\end{align}

\item {$d(t+1)=d(t)-\ell$, assuming $d(t)>\ell\ge1$}.
\\This happens if all arcs from $v_t$ lead to already visited vertices,
and so do all future arcs from the $\ell$ nearest ancestors of $v_t$,
while the $(\ell+1)$th ancestor has at least one future arc leading to an
unvisited vertex. 
The argument in (ii) generalizes and shows that this has probability
\begin{align}\label{iii}
\Bigpar{1-\frac{(1-t/n)p}{1-pt/n}}^{\ell+1}\frac{(1-t/n)p}{1-pt/n}.
\end{align}

\item {$d(t+1)=d(t)-\ell$, assuming $d(t)=\ell\ge0$}.
\\By the same argument as in (ii) and (iii), except that  the $(\ell+1)$th
ancestor does not exist and we ignore it, we obtain the probability
\begin{align}\label{iv}
\Bigpar{1-\frac{(1-t/n)p}{1-pt/n}}^{\ell+1}. 
\end{align}
\end{romanenumerate}
Note that (iv) is the case when $d(t+1)=0$ and thus $v_{t+1}$ is the root of
a new tree in the depth-first forest.

We can summarize (i)--(iv) in the formula
\begin{align}\label{dt+}
  d(t+1) = \bigpar{d(t) +1-\xi_t}^+,
\end{align}
where 
$\xi_t$ is a random variable, independent of
the history, with the 
distribution
\begin{align}\label{xi}
  \PP(\xi_t=k)=(1-\pi_t)^k\pi_t,
\quad k\ge0,
\end{align}
where
\begin{align}\label{xipi}
\pi_t:=
\frac{(1-t/n)p}{1-pt/n}
=1-\frac{1-p}{1-pt/n}.
\end{align}
In other words, $\xi_t$ has the
geometric distribution $\Ge(\pi_t)$.
Define 
\begin{align}
  \label{td}
\td(t):=\sum_{i=1}^{t-1}(1-\xi_i),
\end{align}
and note that \eqref{td} is a sum of independent random variables.
Then $d(t)$ can be recovered from the simpler process $\td(t)$ as follows.

\begin{lemma}\label{Ldtd}
We have
  \begin{align}\label{dtd}
  d(t)=\td(t) - \min_{1\le j\le t} \td(j),
\qquad 1\le t \le n
.\end{align}
\end{lemma}

\begin{proof}
We use induction on $t$. Evidently, \eqref{dtd} holds for $t=1$ since
$d(1)=\td(1)=0$.

Suppose that   \eqref{dtd} holds for some $t<n$.
Then \eqref{td} yields
\begin{align}\label{adv}
    \td(t+1)=\td(t)+1-\xi_t = d(t)+1-\xi_t+\min_{1\le j\le t} \td(j).
\end{align}
If $d(t)+1-\xi_t\ge0$, then \eqref{adv} shows that
$\td(t+1)\ge\min_{1\le j\le t} \td(j)$, and thus
$\min_{1\le j\le t+1} \td(j)=\min_{1\le j\le t} \td(j)$; furthermore,
$d(t+1)=d(t)+1-\xi_t$ by \eqref{dt+}, and it follows that \eqref{dtd} holds
for $t+1$.

On the other hand, if $d(t)+1-\xi_t<0$, then \eqref{adv} shows that
$\td(t+1)<\min_{1\le j\le t} \td(j)$, and thus
$\min_{1\le j\le t+1} \td(j)=\td(t+1)$. In this case, $d(t+1)=0$ by
\eqref{dt+},
and it follows that \eqref{dtd} holds for $t+1$ in this case too.
\end{proof}

\begin{remark}
Similar formulas have been used for other, related, 
problems with random graphs and trees, where trees have been coded as walks,
see for example \cite[Section 1.3]{Aldous1997}. 
Note that in our case, unlike e.g.\ \cite{Aldous1997},
$\td(t)$ may have negative jumps of arbitrary size.
\end{remark}


\subsection{Main result for depth analysis}

Note first that \eqref{xipi} implies that,
using $\gl=p/(1-p)$,
\begin{align}\label{mut}
  \mu_t:=\E\xi_t=\frac{1-\pi_t}{\pi_t}
=\frac{1-p}{p(1-t/n)}
=\frac{1}{\gl(1-t/n)}
.\end{align}
Hence, 
\eqref{td} implies that the expectation of $\td(t)$ is
\begin{align}\label{Etd1}
  \E \bigsqpar{\td(t)} &
= \sum_{i=1}^{t-1} (1-\E\xi_i)
= \sum_{i=1}^{t-1} (1-\mu_i)
= \sum_{i=1}^{t-1} \Bigpar{1-\frac{1}{\gl(1-i/n)}}
.\end{align}
Let $\gth:=t/n$. We fix a $\gthx<1$ and obtain that,
uniformly for $\gth\le\gthx$, 
\begin{align}\label{Etd}
    \E \bigsqpar{\td(t)} 
= \int_0^t\Bigpar{1-\frac{1}{\gl(1- x/n)}} \dd x + O(1)
=n\tl(\gth)+O(1),
\end{align}
where
\begin{align}\label{tl}
  \tl(\gth)
:=
\int_0^\gth\Bigpar{1-\frac{1}{\gl(1-\xtau)}} \dd \xtau 
=\theta+\gl\qw\log(1-\theta)
.\end{align}
Note that the derivative $\tl'(\gth)=1-\gl\qw/(1-\gth)$ is (strictly)
decreasing on $(0,1)$, \ie, $\tl$ is concave.
Moreover,
if $\gl>1$ (\ie, $p>\frac12$)
(the {supercritical} case),
then $\tl'(0)>0$, and \eqref{tl} shows that 
$\tl(\gth)$ is positive and increasing for $\gth<\gth_0:=1-\gl\qw=(2p-1)/p$.
After the maximum at $\gth_0$, 
$\tl(\gth)$ decreases and tends to $-\infty$ as $\gth\upto1$. Hence, there
exists a $\gth_0<\gth_1<1$ such that $\tl(\gth_1)=0$; we then have $\tl(\gth)>0$
for $0<\gth<\gth_1$ and $\tl(\gth)<0$ for $\gth>\gth_1$.
We will see that in this case the depth-first forest \whp{}
contains a giant tree,
of order and height both linear in $n$, while all other trees are small.

On the other hand, if $\gl\le1$ (\ie, $p\le\frac12$)
(the {subcritical} and {critical} cases), 
then $\tl'(0)\le0$ and $\tl(\gth)$ is negative and
decreasing
for all $\gth\in(0,1)$. 
In this case, we define $\gth_0:=\gth_1:=0$ and note that the properties
just stated for $\tl$ still hold (rather trivially).
We will see that in this case
\whp{} all trees in the depth-first forest are small.

Note that in all cases,
\begin{align}\label{gth0}
  \gth_0:=\bigpar{1-\gl\qw}^+
=
  \begin{cases}
    1-\gl\qw,& \gl>1,
\\
0,&\gl\le1,
  \end{cases}
\end{align}
and that
$\gth_1$ is the largest solution in $[0,1)$ to
\begin{align}\label{gth1}
\log(1-\theta_1)=-\gl\theta_1
.\end{align}

\begin{remark}\label{Rgth1}
The equation \eqref{gth1} may also be written
$1-\gth_1=\exp(-\gl \gth_1)$, 
which shows that 
\begin{align}\label{gth1rhoo}
\gth_1=\rhoo(\gl),   
\end{align}
the survival
probability of a Galton--Watson process with $\Po(\gl)$
offspring distribution defined in \eqref{rhoo}.
\end{remark}

We define $\tlp(\gth):=[\tl(\gth)]^+$. Thus, by \eqref{tl} and the comments
above, 
\begin{align}\label{tlp}
  \tlp(\gth)=
  \begin{cases}
\theta+\gl\qw\log(1-\theta)
,& 0\le\gth\le\gth_1,
\\
0,&\gth_1\le\gth\le1.
  \end{cases}
\end{align}

We can now state one of our main results.

\begin{theorem}\label{T1}
We have
\begin{equation}\label{t1}
  \max_{1\le t\le n} \bigabs{d(t)- n\tlp(t/n)}=\Oll(n\qq)
.\end{equation}
\end{theorem}

\begin{proof}
Since \eqref{td} is a sum of independent random variables, 
$\td(t)-\E\td(t)$ 
($t=1,\dots,n$) 
is a martingale,
and Doob's inequality 
\cite[Theorem 10.9.4]{Gut}
yields, for all $T\le n$,
\begin{align}\label{emma}
    \E \bigsqpar{\max_{t\le T}|\td(t)-\E\td(t)|^2}
\le 4 \E \bigsqpar{|\td(T)-\E\td(T)|^2}
=4 \sum_{i=1}^{T-1} \Var(\xi_i).
\end{align}

As above, fix $\gthx<1$, and assume, as we may, that $\gthx>\gth_1$.
Let $\Tx:=\floor{n\gthx}$, and consider first $t\le\Tx$.
For $i<\Tx$, we have $\Var \xi_i = O(1)$, and thus, for $T=\Tx$, 
the sum in \eqref{emma}
is $O(\Tx)=O(n)$.
Consequently, \eqref{emma} yields
\begin{align}
  \max_{t\le\Tx}\bigabs{\td(t)-\E\td(t)} =\Ollnqq
.\end{align}
Hence, by \eqref{Etd},
\begin{align}\label{m*}
\Mx:=
  \max_{t\le\Tx}\bigabs{\td(t)-n\tl(t/n)} =\Ollnqq
.\end{align}
(Note that $\Tx$ and $\Mx$ depend on the choice of $\gthx$.)
For $t\le\Tx$, the definition of $\Mx$ in \eqref{m*} implies
\begin{align}\label{gal}
  \Bigabs{\min_{1\le j\le t}\td(j)- n\min_{1\le j\le t}\tl(j/n)}\le\Mx.
\end{align}
Moreover,
for $t/n\le \gth_1$, we have 
$\min_{1\le j\le t}\tl(j/n)=O(1/n)$,
while
for $t/n\ge \gth_1$, we have 
$\min_{1\le j\le t}\tl(j/n)=\tl(t/n)$.
Hence, for all $t\le \Tx$,
\begin{align}
  \min_{1\le j\le t}\tl(j/n)=\tl(t/n)-\tlp(t/n) + O(1/n),
\end{align}
and thus, by \eqref{gal},
\begin{align}\label{ew}
  \Bigabs{\min_{1\le j\le t}\td(j)- n\tl(t/n)+n\tlp(t/n)}\le\Mx+O(1/n).
\end{align}
Finally, by \eqref{dtd}, \eqref{m*} and \eqref{ew},
\begin{align}\label{jb}
    \bigabs{d(t)- n\tlp(t/n)}\le2\Mx+O(1/n).
\end{align}
This holds uniformly for $t\le\Tx$, and thus, by \eqref{m*}, 
\begin{equation}\label{jesp}
  \max_{1\le t\le\Tx} \bigabs{d(t)- n\tlp(t/n)}=\Oll(n\qq)
.\end{equation}

It remains to consider $\Tx<t\le n$. Then the argument above does not quite
work, because $\pi_t\downto0$ and thus $\Var\xi_t\upto\infty$ as $t\upto n$.
We therefore modify $\xi_t$.
We define
$\hpi_t:=\max\set{\pi_t,\pi_{\Tx}}$; thus $\hpi_t=\pi_t$ for $t\le\Tx$ and
$\hpi_t>\pi_t$ for $t>\Tx$.
We may then define independent random variables $\hxi_t$ such that
$\hxi_t\sim\Ge(\hpi_t)$ and $\hxi_t\le\xi_t$ for all $t< n$.
(Thus, $\hxi_t=\xi_t$ for $t\le\Tx$.)

In analogy with \eqref{td}--\eqref{dtd}, we further define
\begin{align}\label{htd}
  \htd(t)&:=\sum_{i=1}^{t-1}\bigpar{1-\hxi_i},
\\\label{hd}
\hd(t)&:=\htd(t)-\min_{1\le j\le t}\htd(j)
=\max_{1\le j\le t}\sum_{i=j}^{t-1}\bigpar{1-\hxi_i}.
\end{align}
Since $\hxi_i\le \xi_i$, \eqref{hd} implies that
$\hd(t)\ge d(t)$ for all $t$.

We have $\Var\bigsqpar{\hxi_t}=O(1)$, uniformly for all $t<n$, and thus the
argument above yields
\begin{equation}\label{kasp}
  \max_{1\le t\le n} \bigabs{\hd(t)- n[\htl(t/n)]^+}=\Oll(n\qq)
,\end{equation}
where
\begin{align}\label{htl}
  \htl(\gth)
:=
\int_0^\gth
\min\Bigset{\Bigpar{1-\frac{1}{\gl(1-\xtau)}} ,\Bigpar{1-\frac{1}{\gl(1-\gthx)}}} 
\dd \xtau 
.\end{align}
We have $\htl(\gth)=\tl(\gth)$ for $\gth\le\gthx$, and
for $\gth\ge\gthx$, $\htl(\gth)$ is negative and decreasing 
(since $\gthx>\gth_1$). 
Hence, $[\htl(\gth)]^+=\tlp(\gth)$ for all $0<\gth\le1$.
In particular, $[\htl(\gth)]^+=\tlp(\gth)=0$ for all $\gth\ge\gthx$,
and \eqref{kasp} implies
\begin{align}
  \max_{\Tx<t\le n} \hd(t) = \Ollnqq.
\end{align}
Recalling $0\le d(t)\le\hd(t)$, we thus have
\begin{align}
  \max_{\Tx<t\le n} \bigabs{d(t) -n\tlp(t/n)}
=  \max_{\Tx<t\le n} d(t) 
\le \max_{\Tx<t\le n} \hd(t) 
= \Ollnqq,
\end{align}
which completes the proof.
\end{proof}

\begin{corollary}\label{CH}
The height $\gU$ of the depth-first forest is
\begin{equation}\label{gU}
\gU:=  
\max_{1\le t\le n}d(t)= \gu n+\Oll(n\qq),
\end{equation}
where
\begin{align}\label{gu}
\gu=  \gu(p):=  \tlp(\gth_0)
=
  \begin{cases}
    0, & 0<\gl\le1,
\\
1-\gl\qw-\gl\qw\log\gl,
& \gl>1  
.  \end{cases}
\end{align}
\end{corollary}

\begin{proof}
  Immediate from \refT{T1} and \eqref{tl}, 
since we have $\max_t\tlp(t/n)=\max_\gth\tlp(\gth)+O(1/n)$ and
$\max_\gth\tlp(\gth)=\tlp(\gth_0)=\tl(\gth_0)$.
\end{proof}

In \refSS{SSasn} we will improve this when $\gl>1$,
and show that then
the height $\gU$ is asymptotically normally
distributed (\refT{TH2}).

\begin{corollary}\label{CA}
  The average depth $\bd$ in the depth-first forest is
  \begin{align}\label{ca}
\bd:=
    \frac{1}n\sum_{t=1}^n d(t) = \ga n + \Oll(n\qq),
  \end{align}
where
$\ga=0$ if $\gl\le1$, and, in general,
\begin{align}\label{ga}
  \ga=\ga(p):=
\frac12\gth_1^2 -\frac{1}{\gl}\Bigpar{(1-\gth_1)\log(1-\gth_1)+\gth_1}
=\frac{\gl-1}{\gl}\,\gth_1-\frac12\gth_1^2
.\end{align}
\end{corollary}

\begin{proof}
  By \refT{T1},
  \begin{align}\label{bal1}
\frac{1}{n}\sum_{t=1}^n  d(t) = \sum_{t=1}^n \tlp(t/n) + \Oll\bigpar{n\qq} 
=n \ga + \Oll\bigpar{n\qq}, 
  \end{align}
where
\begin{align}\label{bal2}
  \ga &:= \int_0^1 \tlp(\xtau)\dd\xtau
=
\int_0^{\gth_1} \tl(\xtau)\dd\xtau
=\int_0^{\gth_1}\Bigpar{\xtau+\gl\qw\log(1-\xtau)}\dd\xtau
\notag\\&\phantom:
=\frac12\gth_1^2 -\gl\qw\Bigpar{(1-\gth_1)\log(1-\gth_1)+\gth_1}
,\end{align}
which yields \eqref{ga}, using \eqref{gth1}.
\end{proof}

\begin{remark}\label{Rslow}
When $\gl>1$, the height $\gU$ and average depth $\bd$
are thus linear in $n$, unlike many other
types of random trees.
This might imply a rather slow 
performance of algorithms that operate on the depth-first forest
if it is built explicitly in  a computer's memory.
\end{remark}

\subsection{Asymptotic normality}\label{SSasn}
In this subsection, we show that in the supercritical case $\gl>1$,
\refT{T1} can be improved to yield convergence of $d(t)$ (after rescaling)
to a Gaussian process, at least on $[0,\gth_1)$.
As a consequence, we show that the height $\gU$ is asymptotically normal.

Recall that for an interval $I\subseteq\bbR$, $D(I)$ is the  space
of functions $I\to\bbR$
that are right-continuous with left limits (\emph{c\`adl\`ag})
equipped with the Skorohod topology.
For definitions of the topology
see \eg{} \cite{Billingsley}, 
\cite{JS}, 
\cite[Appendix A.2]{Kallenberg}, 
or \cite{SJ94};
for our purposes it is enough to know that convergence in $D(I)$ to a
continuous limit is equivalent to uniform convergence on compact subsets of
$I$. (Note that it thus matters if the endpoints are included in $I$ or not;
for example, convergence in $D\oio$ and $D\oi$ mean different things.)

We define $d(0):=\td(0):=0$.

\begin{lemma}  \label{LN1}
Assume $\gl>1$.
Then
\begin{align}\label{ln1a}
  n\qqw\bigpar{\td(\floor{n\gth}) -n\tl(\gth)}\dto Z(\gth)
\qquad\text{in $D[0,1)$} 
,\end{align}
where $Z(\gth)$ is a continuous Gaussian process on $\oio$
with mean $\E Z(\gth)=0$ and covariance
$\Cov\bigpar{Z(x),Z(y)}=\gssx\bigpar{\min\set{x,y}}$, where
\begin{align}\label{ln1b}
  \gssx(\gth):=
\frac{(1-p)^2\gth}{p^2(1-\gth)}-\frac{1-p}{p}\log(1-\gth)
=\gl\qww\frac{\gth}{1-\gth}-\gl\qw\log(1-\gth).
\end{align}
Equivalently, $Z(\gth)=B\bigpar{\gssx(\gth)}$ for a 
Brownian motion $B(x)$.
\end{lemma}

\begin{proof}
Since the random variables $\xi_t$ are independent, \eqref{td} 
and \eqref{xi}--\eqref{xipi} yield, similarly to \eqref{Etd1}, 
\begin{align}\label{v1}
\Var\bigsqpar{\td(t)} &
=\sum_{i=1}^{t-1}\Var \xi_i   
=\sum_{i=1}^{t-1}\frac{1-\pi_i}{\pi_i^2}
=\sum_{i=1}^{t-1}\frac{(1-p)(1-pi/n)}{p^2(1-i/n)^2}
.\end{align}
Hence, uniformly for $t/n\le\gthx$ for any $\gthx<1$, 
\begin{align}\label{v2}
  \Var\bigsqpar{\td(t)}
= n\gssx(t/n)+O(1),
\end{align}
with
\begin{align}\label{v3}
 \gssx(\gth):= 
\frac{1-p}{p^2}\int_0^{\gth}\frac{1-px}{(1-x)^2}\dd x
=\frac{(1-p)^2\gth}{p^2(1-\gth)}-\frac{1-p}{p}\log(1-\gth),
\end{align}
in agreement with \eqref{ln1b}.
Since also $\E\td\bigpar{\floor{n\gth}}=n\tl(\gth)+O(1)$ by \eqref{Etd},
the marginal convergence for a fixed $\gth$ in \eqref{ln1a} follows by the
classical central limit theorem for independent (not identically
distributed) variables, \eg{} using Lyapounov's condition
\cite[Theorem 7.2.2]{Gut}.

The functional limit \eqref{ln1a} is thus a version of Donsker's theorem 
\cite[Theorem 16.1]{Billingsley}, extended from the i.i.d.\ case to
the non-identically distributed variables $\xi_i$.
We expect that such generalizations of Donsker's theorem exist in the
literature, but we do not know any specific reference;
however, standard proofs extend without problems.
(For example, the proof of \cite[Theorem 16.1]{Billingsley}.
Alternatively, finite-dimensional convergence follows from the classical
central limit theorem, and  tightness in $D\oio$ can easily be shown 
by \eg{}
Aldous's criterion \cite[Theorem 16.11]{Kallenberg}.)
Moreover, \eqref{ln1a} follows also directly
by general results on convergence of
martingales, for example
\cite[Proposition 2.6]{SJ94} 
or \cite[Proposition 9.1]{SJ154}, which both are based on
\cite[Theorem VIII.3.12]{JS}.
\end{proof}

\begin{lemma}\label{LN2}
Assume $\gl>1$ and let $0<\gthx<\gth_1$.
Then 
\begin{align}\label{ln2}
  \min_{1\le j\le \floor{n\gthx}}\td(j) =\op\bigpar{n\qq}.
\end{align}
\end{lemma}

\begin{proof}
  Let $t_n:=\ceil{n^{2/3}}$.
If $n$ is large enough, then $t_n<n\gthx$, and, since $\tl'(0)=1-\gl\qw>0$
by \eqref{tl},
\begin{align}\label{ln2a}
  \min_{t_n/n\le\gth\le\gthx}\tl(\gth)
=\tl(t_n/n) \ge ct_n/n \ge c n^{-1/3}
\end{align}
for some constant $c>0$.
Furthermore, \eqref{m*} implies
\begin{align}\label{ln2b}
  \max_{t_n\le t\le n\gthx} \bigabs{\td(t)- n\tl(t/n)}=\Ollnqq = \op\bigpar{n^{2/3}}
\end{align}
(recall that $\Oll(a_n)$ implies $\op(\go_na_n)$ for any $a_n$ and 
any $\go_n\to\infty$).
It follows from \eqref{ln2a}--\eqref{ln2b} that \whp{} 
$ \bigabs{\td(t)- n\tl(t/n)} < n\tl(t/n)$ for all $t\in[t_n,n\gthx]$.
Hence, \whp, $\td(t)>0=\td(1)$ for  all $t\in[t_n,n\gthx]$.
Consequently, \whp,
\begin{align}\label{minmin}
  \min_{1\le t\le n\gthx}\td(t)=\min_{1\le t\le t_n}\td(t).  
\end{align}

For  $t\le t_n$, we use Doob's inequality in the form \eqref{emma} again.
Since $\min_{t\le t_n}\td(t)\le\td(1)=0$ and $\E\td(t)\ge0$ for $t\le t_n$
(for $n$ large),
we have $\bigabs{  \min_{t\le t_n}\td(t)}\le \max_{t\le t_n}|\td(t)-\E\td(t)|$
and thus \eqref{emma} yields
\begin{align}\label{ln2c}
\E\bigabs{  \min_{t\le t_n}\td(t)}^2
\le4 \sum_{i=1}^{t_n-1} \Var(\xi_i)
=O(t_n) = O\bigpar{n^{2/3}}.
\end{align}
Hence,
\begin{align}\label{ln2d}
   \min_{t\le t_n}\td(t) = \Oll\bigpar{n^{1/3}} = \op\bigpar{n\qq}.
\end{align}
The proof is completed by combining \eqref{minmin} and \eqref{ln2d}.
\end{proof}

\begin{theorem}\label{TN1}
  Assume $\gl>1$.
Then
\begin{align}\label{tn1}
  n\qqw\bigpar{d(\floor{n\gth}) -n\tl(\gth)}\dto Z(\gth)
\qquad\text{in $D[0,\gth_1)$} 
\end{align}
where $Z(\gth)$ is the continuous Gaussian process defined in \refL{LN1}.
\end{theorem}

\begin{proof}
  By \eqref{dtd} and \refL{LN2},
for any $\gthx<\gth_1$,
\begin{align}
\max_{0\le t\le\floor{n\gthx}} \bigabs{ d(t)-\td(t)}
=\Bigabs{  \min_{1\le j\le \floor{n\gthx}}\td(j)}
=\op\bigpar{n\qq}
.\end{align}
The theorem now follows
from \refL{LN1}.
\end{proof}

\begin{theorem}\label{TH2}
  Let $\gl>1$. Then the height $\gU$ of the \dff{} has an asymptotic normal
  distribution:
  \begin{align}\label{th2a}
\frac{\gU-\gu n}{\sqrt n}\dto N\bigpar{0,\gss}    
  \end{align}
with $\gu$ given by \eqref{gu}, and
\begin{align}\label{th2b}
  \gss:= \gl\qw-\gl\qww+\gl\qw\log\gl.
\end{align}
\end{theorem}

\begin{proof}
Fix some $\gthx\in(\gth_0,\gth_1)$. 
  By \refT{TN1} and
the Skorohod coupling theorem \cite[Theorem~4.30]{Kallenberg},
we may assume that the random variables for different $n$ are coupled such
that \eqref{tn1} holds (almost) surely. 
Since $Z(\gth)$ is continuous, 
this implies uniform covergence on $[0,\gthx]$, \ie,
\begin{align}\label{sam1}
  d\bigpar{\floor{n\gth}}=n\tl(\gth)+n\qq Z(\gth) + o\bigpar{n\qq},
\end{align}
uniformly on  $[0,\gthx]$. 
(The $o(n\qq)$ here are random, but uniform in $\gth$.)
For $|\gth-\gth_0|\le n^{-1/6}$, we have $Z(\gth)=Z(\gth_0)+o(1)$, since
$Z$ is continuous, and thus \eqref{sam1} yields, almost surely,
\begin{align}\label{sam2}
  d\bigpar{\floor{n\gth}}=n\tl(\gth)+n\qq Z(\gth_0) + o\bigpar{n\qq},
\qquad |\gth-\gth_0|\le n^{-1/6}.
\end{align}
Since $\max_\gth \tl(\gth)=\tl(\gth_0)$, it follows that
\begin{align}\label{sam3}
\max_{ |\gth-\gth_0|\le n^{-1/6}} d\bigpar{\floor{n\gth}}
=n\tl(\gth_0)+n\qq Z(\gth_0) + o\bigpar{n\qq}.
\end{align}

On the other hand, for $|\gth-\gth_0|\ge n^{-1/6}$, we have by a Taylor
expansion, for some $c>0$,
\begin{align}
  \tl(\gth)\le \tl(\gth_0) - c (\gth-\gth_0)^2
\le \tl(\gth_0) - c n^{-1/3}.
\end{align}
Hence, \eqref{t1} implies
\begin{equation}\label{sam4}
  \max_{|\gth-\gth_0|\ge n^{-1/6}} d(t)
\le n \max_{|\gth-\gth_0|\ge n^{-1/6}} \tl(\gth) + \Ollnqq
\le n \tl(\gth_0)-c n^{2/3} + \Ollnqq
.\end{equation}
Comparing \eqref{sam3} and \eqref{sam4}, we see that \whp{} the maximum in
\eqref{sam3} is larger than the one in \eqref{sam4}, and thus
\begin{align}
  \gU=\max_{0\le\gth\le1}d\bigpar{\floor{n\gth}} 
=n\tl(\gth_0)+n\qq Z(\gth_0) + o\bigpar{n\qq}.
\end{align}
Hence, \whp,
\begin{align}
  \frac{\gU-\tl(\gth_0) n}{\sqrt n} = Z(\gth_0)+o(1),
\end{align}
which implies
\begin{align}
 \frac{\gU-\tl(\gth_0) n}{\sqrt n} \dto Z(\gth_0) \sim N\bigpar{0,\gssx(\gth_0)}.
\end{align}
Since $\tl(\gth_0)=\gu$ by \eqref{gu}, 
this shows \eqref{th2a} with
$\gss:=\gssx(\gth_0)$, which gives \eqref{th2b} by \eqref{ln1b} and
$\gth_0:=1-\gl\qw$. 
\end{proof}

\subsection{The trees in the forest}

\begin{theorem}\label{TT}
  Let $N$ be the number of trees in the depth-first forest.
Then
\begin{align}\label{tt}
  N=\psi n + \Ollnqq,
\end{align}
where
\begin{align}\label{psi}
  \psi=\psi(p):=
1-\gth_1-\frac{\gl}{2}(1-\gth_1)^2.
\end{align}
\end{theorem}

Figure~\ref{fig:psilambda} shows the parameter $\psi$ as a function of the
average degree $\lambda$.
\begin{figure}[ht]
\includegraphics[height=7cm]{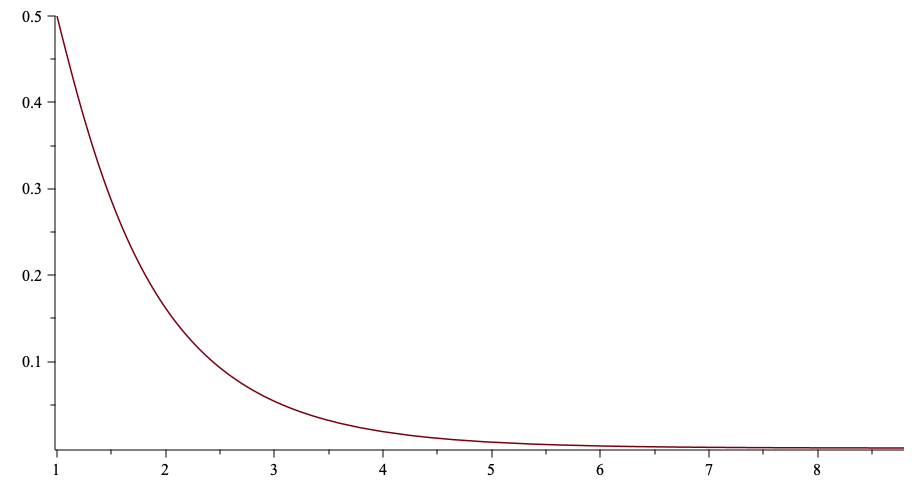}
\caption{$\psi$, as function of $\lambda$.}
\label{fig:psilambda}
\end{figure}

\begin{proof}
  Let $\JJ_t:=\indic{d(t)=0}$,
the indicator that vertex $t$ is a root and thus starts a new tree.
Thus $N=\sum_1^n \JJ_t$.

If $\gth_1>0$ (\ie, $\gl>1$), then \refT{T1} shows that \whp\
$d(t)>0$ in the interval $(1,n\gth_1)$, except possibly close to the endpoints.
Thus the DFS will find one giant tree of order $\approx \gth_1 n$, possibly
preceded by a few small trees, and, as we will see later in the proof, 
followed by many
small trees. To obtain a precise estimate, we note that there exists a
constant $c>0$ such that $\tl(\gth)\ge\min\set{c\gth,c(\gth_1-\gth)}$ for
$\gth\in[0,\gth_1]$. 
Hence, 
if $t\le n\gth_1$ and $d(t)=0$, 
then $\td(t)\le d(t)=0$ by \eqref{dtd} and,
recalling \eqref{m*},  
\begin{align}
\Mx\ge n\tl(t/n) \ge c\min\set{t,n\gth_1-t}.  
\end{align}
Consequently, $d(t)=0$ with $t\le n\gth_1$ implies 
$t\in[1,c\qw \Mx] \cup [n\gth_1-c\qw \Mx,n\gth_1]$.
The number of such $t$ is thus $O(\Mx+1)=\Ollnqq$,
using \eqref{m*}.

Let $T_1:=\ceil{n\gth_1}$.
We have just shown that (the case $\gth_1=0$ is trivial)
\begin{align}\label{mma}
  \sum_{t=1}^{T_1-1} \JJ_t = \Ollnqq.
\end{align}

It remains to consider $t\ge T_1$.
For any integer $k\ge0$, the conditional distribution of $\xi_t-k$ given
$\xi_t\ge k$ equals the distribution of $\xi_t$. Hence,
recalling \eqref{mut},
\begin{align}\label{erika}
  \E\bigsqpar{(\xi_t-k)^+}
=\E\bigsqpar{\xi_t-k\mid\xi_t\ge k}\PP(\xi_t\ge k)
=\mu_t\PP(\xi_t-k\ge0).
\end{align}
We use again the stochastic recursion \eqref{dt+}.
Let $\cF_t$ be the $\gs$-field generated by $\xi_1,\dots,\xi_{t-1}$.
Then $d(t)$ is $\cF_t$-measurable, while $\xi_t$ is independent of $\cF_t$.
Hence, \eqref{dt+} and \eqref{erika} yield
\begin{align}\label{ele}
  \E\bigsqpar{d(t+1)\mid\cF_t}&
= \E\bigsqpar{d(t)+1-\xi_t\mid\cF_t}
+ \E\bigsqpar{\xpar{\xi_t-1-d(t)}^+\mid\cF_t}
\notag\\&
=d(t)+1-\mu_t+\mu_t\PP\bigsqpar{\xi_t-1-d(t)\ge0\mid\cF_t}
\notag\\&
=d(t)+1-\mu_t+\mu_t\PP\bigsqpar{d(t+1)=0\mid\cF_t}
\notag\\&
=d(t)+1-\mu_t+\mu_t\E\bigsqpar{\JJ_{t+1}\mid\cF_t}
.\end{align}
We write $\gD d(t):=d(t+1)-d(t)$ and $\bJ_t:=1-\JJ_t$. Then \eqref{ele} yields
\begin{align}\label{win}
  \E\bigsqpar{\gD d(t) -1 + \mu_t\bJ_{t+1}\mid\cF_t}=0.
\end{align}
Define
\begin{align}\label{cM}
  \cM_t:
=\sum_{i=1}^{t-1}\mu_i\qw\bigpar{\gD d(i)-1+\mu_i\bJ_{i+1}}
=\sum_{i=1}^{t-1}\Bigpar{\mu_i\qw\gD d(i)-\mu_i\qw+\bJ_{i+1}}
.\end{align}
Then $\cM_t$ is $\cF_t$-measurable, and \eqref{win} shows that $\cM_t$ is a
martingale.
We have, with $\gD \cM_t:=\cM_{t+1}-\cM_t$, using \eqref{dt+},
\begin{align}\label{cm1}
  |\gD\cM_t|\le\mu_t\qw\bigabs{d(t+1)-d(t)-1}+\bJ_{t+1}
\le\mu_t\qw\xi_t+1,
\end{align}
and thus, since $\pi_t\le p<1$ for all $t$ by \eqref{xipi},
\begin{align}\label{cm2}
  \E\abs{\gD\cM_t}^2
\le 2\mu_t\qww\E\xi_t^2+2
=2\Bigpar{\frac{\pi_t}{1-\pi_t}}^2\frac{1-\pi_t+(1-\pi_t)^2}{\pi_t^2}+2=O(1).
\end{align}
Hence, uniformly for all $T\le n$,
\begin{align}\label{cm3}
  \E \cM_T^2
=\sum_{t=1}^{T-1}\E|\gD\cM_t|^2=O(T)=O(n).
\end{align}

The definition \eqref{cM} yields
\begin{align}  \label{cm4}
\cM_n-\cM_{T_1}&
=
\sum_{t=T_1}^{n-1}\mu_t\qw\gD d(t) - \sum_{t=T_1}^{n-1}\mu_t\qw
+\sum_{t=T_1}^{n-1}\bJ_{t+1}
.\end{align}
By a summation by parts, and 
interpreting $\mu_n\qw:=0$,
\begin{align}  \label{cm5}
\sum_{t=T_1}^{n-1}\mu_t\qw\gD d(t)
=\sum_{t=T_1+1}^{n}\bigpar{\mu_{t-1}\qw-\mu_t\qw} d(t) -\mu_{T_1}\qw d(T_1)
.\end{align}
As $t$ increases,  $\mu_t$ increases by \eqref{mut}, and thus 
$\mu_{t-1}\qw-\mu_t\qw>0$. Hence, \eqref{cm5} implies
\begin{align}  \label{cm6}
\Bigabs{\sum_{t=T_1}^{n-1}\mu_t\qw\gD d(t)}
&
\le\sum_{t=T_1+1}^{n}\bigpar{\mu_{t-1}\qw-\mu_t\qw}\sup_{i> T_1}|d(t)| 
+\mu_{T_1}\qw |d(T_1)|
\le 2 \mu_{T_1}\qw \sup_{i\ge T_1}|d(t)| 
\notag\\&
=\Ollnqq
\end{align}
by \eqref{t1}, since $\tlp(t/n)=0$ for $t\ge T_1\ge n\gth_1$.
Furthermore, \eqref{cm3} shows that $\cM_n,\cM_{T_1}=\Ollnqq$.
Hence, \eqref{cm4} yields, using \eqref{mut},
\begin{align}\label{cm7}
  \sum_{t=T_1+1}^n \JJ_t&
= n-T_1 -  \sum_{t=T_1+1}^n \bJ_t
=n-T_1-  \sum_{t=T_1}^{n-1}\mu_t\qw + \Ollnqq
\notag\\&
=n-T_1-  \sum_{t=T_1}^{n-1}\gl(1-t/n) + \Ollnqq
=n\psi+\Ollnqq,
\end{align}
where
\begin{align}\label{tau=}
  \psi:=
1-\gth_1 -\int_{\gth_1}^1\gl{(1-\xtau)}\dd\xtau
=1-\gth_1-\frac{\gl}{2}(1-\gth_1)^2.
\end{align}
The result follows by \eqref{cm7} and \eqref{mma}.
\end{proof}

The arguments in the proof of \refT{TT} show
that in the supercritical case $\gl>1$, the DFS \whp{}
find first possibly a few small trees, then a giant tree containing all
$v_t$ with
$\Ollnqq\le t\le \gth_1n+\Ollnqq$, and then a large number of small trees.
We give some details in the following lemma and theorem.

\begin{lemma}\label{Lroots}
Let $(a,b)$ be a fixed interval with $0\le a<b\le 1$ and $b>\gth_1$.
Then \whp{} there exists a root $v_t$ in the \dff{}  with $t/n\in(a,b)$.
\end{lemma}

\begin{proof}
  By increasing $a$,
we may assume that
$\gth_1<a<b\le1$.
Then, cf.\ \eqref{tl}, 
$\tl'(a)<\tl'(\gth_1)\le0$ and thus $ \gl(1-a)<1$.
Hence, the  argument yielding \eqref{cm7} in the proof of \refT{TT} 
yields also
\begin{align}\label{cm7ab}
  \sum_{t=\ceil{an}}^{\floor{bn}} \JJ_t&
=bn-an-  \sum_{t=\ceil{an}}^{\floor{bn}}\gl(1-t/n) + \Ollnqq
=cn+\Ollnqq,
\end{align}
with $c>(b-a)(1-\gl(1-a))>0$.
Hence, \whp{} there are many roots $v_t$ with $t\in(an,bn)$.
\end{proof}

\begin{theorem}\label{TT2}
Let $\TTx_1$ be the largest tree in the depth-first forest.
\begin{romanenumerate}
  \item\label{TT2<=} 
If\/ $\gl\le1$, then $|\TTx_1|=\op(n)$.
  \item \label{TT2>}
If\/ $\gl>1$, then 
$|\TTx_1|=\gth_1 n+\Ollnqq$.
Furthermore, the second largest tree has order $|\TTx_2|=\op(n)$.
\end{romanenumerate}
\end{theorem}
\begin{proof}
Let $\eps>0$.
By covering $[\gth_1,1]$ with a finite number of 
intervals of length $<\eps/2$, it follows from \refL{Lroots} that
\whp{} every tree $\TTx$ having a root $v_t$ with $t>(\gth_1-\eps/2)n$
has $|\TTx|\le\eps n$.

In particular, if $\gl\le1$, so $\gth_1=0$, this applies to all trees, and
thus \whp{} $|\TTx_1|\le\eps n$, which proves \ref{TT2<=}. 

Suppose now $\gl>1$.
  Consider the tree $\TTx$ in the \dff{} that contains $v_{\floor{n\gth_0}}$,
denote its root by $v_r$ and let $v_s$ be its last vertex.
By the proof of \refT{TT}, $d(t)>0$ for $\Ollnqq \le t\le \gth_1 n-\Ollnqq$,
and thus $r=\Ollnqq$ and $s\ge \gth_1n-\Ollnqq$.

On the other hand, let $\gthx\in(\gth_1,1)$. 
If $s\ge\gth_1n$, let $u:=\min\set{s,\floor{\gthx n}}$.
Since $r/n\le\gth_0$, we have $\tl(r/n)\ge0$.
Furthermore, \eqref{dtd} implies that  
$\td(t)>\min_{j\le t}\td(j)=\td(r)$ for $t\in(r,s]$, 
and thus $\td(u)>\td(r)$. 
Hence, by \eqref{m*},
\begin{align}\label{tt2a}
  -n\tl(u/n) 
\le n\tl(r/n)-n\tl(u/n) 
\le 2\Mx + \td(r)-\td(u)
\le 2\Mx
=\Ollnqq.
\end{align}
Since $\tl(\gth_1)=0$ and $\tl'(\gth)\le -c<0$ for $\gth\ge\gth_1$,
it follows that $u\le \gth_1n+\Ollnqq$, and thus
$s\le \gth_1n+\Ollnqq$. 

Consequently, $s= \gth_1n+\Ollnqq$, and thus $|\TTx|=s-r+1=\gth_1n+\Ollnqq$.
Furthermore, any tree found before $\TTx$ has order $\le r=\op(n)$.
The first part of the proof now shows that for every $\eps>0$, there
is \whp{} no tree other than $|\TTx|$ of order $>\eps n$.
Hence, \whp{} $\TTx$ is the largest tree, and \ref{TT2>} follows.
\end{proof}

\begin{remark}\label{RGW}
As said in \refR{Rgth1},  $\gth_1$,
the asymptotic fraction of vertices in the  giant tree, 
equals the survival probability of a Galton--Watson process with $\Po(\gl)$
offspring distribution.
Heuristically, this may be explained by the following argument, well known
from similar situations. Start at a random vertex
and follow the arcs backwards. 
The indegree of a given vertex is asymptotically $\Po(\gl)$, and 
the process of exploring backwards from a vertex may be approximated by
a Galton--Watson process with this offspring distribution.
Hence, the probability of a ``large'' backwards process converges to the
survival probability $\gth_1$. 
It seems reasonable that most vertices in the giant tree have a large
backwards process, while most vertices outside the giant have a small
backwards process.

Note also that the asymptotic size of the giant tree thus equals the
asymptotic size of the giant component in an undirected Erd{\H o}s--R{\'e}nyi
random graph $G(n,\gl/n)$, which heuristically is given by the same argument.
(See also \refR{RER} and \cite{SJ364-general}.) 
\end{remark}

\subsection{Types of arcs}\label{SSGeoTypes}

Recall from the introduction the classification of the arcs in the digraph $G$. 
Since we assume that the outdegrees are $\Ge(1-p)$ and independent, the total
number of arcs, $M$ say, has a negative binomial distribution
with mean $\gl n$,
and, by a weak version of the law of large numbers, 
\begin{align}\label{sw}
  M=\gl n + \Ollnqq.
\end{align}

In the following theorem, we give the asymptotics of the number of
arcs of each type. 

\begin{theorem}\label{TGeoA}
  Let $L$, $T$, $B$, $F$ and $C$ be the numbers of loops, tree arcs, back
  arcs, forward arcs, and cross arcs in the random digraph.
Then
\begin{align}
L&=  \Oll(1), \label{tgeoal}
\\
T &= \chit n + \Ollnqq, \label{tgeoat}
\\
 B&= \gb n + \Ollnqq, \label{tgeoab}
\\
 F&=\gf n + \Ollnqq,\label{tgeoaf}
\\
 C &= \chix n + \Ollnqq, \label{tgeoac}
\end{align}
where
\begin{align}
\label{ggeot}
\chit&:=
\chix:=1-\psi 
=\gth_1+\frac{\gl}{2}(1-\gth_1)^2,
\\\label{ggeob}
  \gb &:= 
\gf:=\gl\ga 
=(\gl-1)\gth_1-\frac{\gl}{2}\gth_1^2
.\end{align}
\end{theorem}

\begin{proof}
Let $\eta_t$ be the number of arcs from $v_t$, and let
$\etam_t,\etal_t,\etap_t$ be the numbers of these arcs that lead to some
$v_u$ with $u<t$, $u=t$ and $u>t$, respectively.
Then 
\begin{align}\label{ta1}
  L=\sumtn\etal_t.
\end{align}
%
Furthermore, an arc $v_tv_u$ with $u>t$ is either a tree arc or a forward
arc; conversely, every tree arc or forward arc is of this type.
Consequently,
\begin{align}\label{tatf}
  T+F&=\sumtn\etap_t.
\end{align}
Similarly, or by \eqref{ta1} and \eqref{tatf},
\begin{align}\label{tabc}
B+C&=\sumtn\etam_t.
\end{align}
Conditioned on $\eta_t$, $\etam_t$ has a binomial distribution
$\Bin(\eta_t,(t-1)/n)$,
since each arc has probability $(t-1)/n$ to go to a vertex $v_u$ with $u<t$.
In general, if $X\sim\Bin(m,p)$, then $\E X = mp$ and
$\E X^2=\Var X + (\E X)^2=mp(1-p)+(m p)^2$. Hence,
by first conditioning on $\eta_t$,
\begin{align}\label{ta2}
  \E\etam_t &= \E\Bigsqpar{\eta_t\frac{t-1}n}
=\gl \frac{t-1}n,
\\\label{ta3}
\Var\etam_t &\le \E \bigsqpar{(\etam_t)^2}
=\E \Bigsqpar{\eta_t \frac{t-1}{n}\Bigpar{1-\frac{t-1}n}
+\eta_t^2\Bigpar{\frac{t-1}{n}}^2}
= O(1).
\end{align}
Furthermore, the random variables $\etam_t$, $t=1,\dots,n$, are independent.
Hence, \eqref{tabc} yields
\begin{align}\label{tabc1}
  \E B + \E C &= \sumtn\E\etam_t = \gl\sumtn\frac{t-1}{n} 
=\frac{\gl}2(n-1),
\\\label{tabc2}
\Var\sqpar{B + C}& = \sumtn \Var\etam_t = O(n),
\end{align}
and thus
\begin{align}\label{tabc3}
  B+C 
=\E\sqpar{B+C}+\Ollnqq
= \frac{\gl}2n+\Ollnqq.
\end{align}

The same argument with \eqref{tatf} yields
\begin{align}\label{tatf1}
  \E T + \E F & 
=\frac{\gl}2(n-1),
\\\label{tatf2}
\Var\sqpar{F + T}&  
= O(n),
\\\label{tatf3}
  F+T &= \frac{\gl}2n+\Ollnqq.
\end{align}
Similarly, conditioned on $\eta_t$, we have $\etal\sim\Bin(\eta_t,1/n)$,
and we find 
\begin{align}\label{tal1}
  \E L &= \sumtn\E\etal_t = n\gl\frac{1}{n}
=\gl,
\\\label{tal2}
\Var{L}& = \sumtn \Var\etal_t = O(1),
\\\label{tal3}
  L& = \gl+\Oll(1) = \Oll(1).
\end{align}
This proves \eqref{tgeoal}. We prove \eqref{tgeoat}--\eqref{tgeoac} one by one.

\pfitem{$T$}
In any forest, the number of vertices equals the number of edges + the
number of trees.
Hence, $T=n-N$, where $n$ is the number of trees in the \dff, and thus 
\refT{TT} implies \eqref{tgeoat} with $\tau=1-\psi$ given by \eqref{ggeot}.

\pfitem{$B$}  
Let $B_t$ be the number of back arcs from $v_t$; thus $B=\sum_1^n B_t$.
Let $\cF_t$ be the $\gs$-field generated by all arcs from $v_i$, $i\le t$
(\ie, by the outdegrees $\eta_i$ and the endpoints of all these arcs);
note that this includes complete information on the DFS until $v_{t+1}$ is
found, but also on some further arcs (the future arcs from the ancestors of
$v_{t+1}$).
Then $d(t)$ is $\cF_{t-1}$-measurable and $B_t$ is $\cF_t$-measurable.
Moreover, $\eta_t$ is independent of $\cF_{t-1}$. Thus, conditioned on
$\cF_{t-1}$, we still have $\eta_t\sim\PQ$; we also know $d(t)$, and each arc
from $v_t$ is a back arc with probability $d(t)/n$.
Hence,
$\E\bigsqpar{B_t\mid\cF_{t-1},\eta_t}=\eta_t d(t)/n$, and consequently
\begin{align}\label{tb1}
  \E\bigsqpar{B_t\mid\cF_{t-1}}
=\E\Bigsqpar{\eta_t\frac{d(t)}n\mid\cF_{t-1}}
=\gl\frac{d(t)}n.
\end{align}
Similarly, since, as above, $X\sim\Bin(m,p)$ implies 
$\E X^2=mp(1-p)+(m p)^2$,
\begin{align}\label{tb2}
  \E\bigsqpar{B_t^2\mid\cF_{t-1}}
=\E\Bigsqpar{\eta_t\frac{d(t)}n\Bigpar{1-\frac{d(t)}{n}}
  +\Bigpar{\eta_t\frac{d(t)}n}^2\mid\cF_{t-1}}
=O(1).
\end{align}

Define $\gD Z_t:=B_t-\gl d(t)/n$ and $Z_t:=\sum_1^t \gD Z_i$.
Then \eqref{tb1} shows that $\E\bigsqpar{\gD Z_t\mid\cF_t}=0$, and thus
$(Z_i)_0^n$ is a martingale, with $Z_0=0$. Hence, $\E Z_n=0$.
Furthermore, \eqref{tb1} implies
$\E\bigsqpar{(\gD Z_t)^2}=\E\bigsqpar{(B_t-\E[B_t\mid\cF_{t-1}])^2}
\le \E\bigsqpar{B_t^2}$, 
and thus by \eqref{tb2},
\begin{align}
  \E \bigsqpar{Z_n^2}
= \Var\bigsqpar{Z_n}
=\sum_{t=1}^n \E\bigsqpar{(\gD Z_t)^2} 
\le \sum_{t=1}^n \E\bigsqpar{B_t^2} 
= O(n).
\end{align}
Consequently, $Z_n=\Ollnqq$, and thus
\begin{align}\label{tb4}
  B=
\sum_{t=1}^n B_t 
=Z_n+\sum_{t=1}^n \gl\frac{d(t)}n
=\gl\bd + Z_n
=\gl\bd+\Ollnqq.
\end{align}
Finally, \eqref{tb4} and \refC{CA} yield
\begin{align}\label{babb}
  B=\gl\bd + \Ollnqq
=\gl\ga n+\Ollnqq,
\end{align}
which shows \eqref{tgeoab}
with $\gb=\gl\ga$ as in \eqref{ggeob}, recalling \eqref{ga}.

\pfitem{$F$}
 By \eqref{tatf3} and \eqref{tgeoat},
we have \eqref{tgeoaf}
with $\gf=\frac{\gl}2-\chit$, which agrees with \eqref{ggeob} by \eqref{ggeot}
and a simple calculation. In particular, $\gf=\gb$.

\pfitem{$C$}
Similarly,
it follows from \eqref{tabc3} and \eqref{tgeoab} that we have \eqref{tgeoac}
with $\chix:=\gl/2-\gb$.
Since we have found $\gb=\gf$, and we always have $\chit+\gf=\gl/2=\gb+\chix$,
see \eqref{tabc3} and \eqref{tatf3}, 
we thus have $\chix=\chit$, and thus \eqref{ggeot} holds.
\end{proof}

Note that $T+F$ and $B+C$ are asymptotically normal;
this follows immediately from \eqref{tatf} and \eqref{tabc}
by the central limit theorem.

\begin{conjecture}
 All four variables $T,B,F,C$ are (jointly) asymptotically
normal.
\end{conjecture}

The equalitites $\chit=\chix$ and $\gb=\gf$ mean asymptotic equality of the
corresponding expectations of  numbers of arcs.
In fact, there are exact equalities.

\begin{theorem}\label{TT=C}
  For any $n$,
$\E T = \E C$ and\/ $\E B = \E F=\gl\E\bd$.
\end{theorem}

\begin{proof}
  Let $v,w$ be two distinct vertices.
If the DFS finds $w$ as a descendant of $v$, then there will later
be $\Ge(1-p)$ arcs from $w$, and each has probability $1/n$ of being a back
arc to $v$. Similarly, there will be $\Ge(1-p)$ future arcs from $v$, and each
has probability $1/n$ of being a forward arc to $w$.
Hence, if $I_{vw}$ is the indicator that $w$ is a descendant of $v$, 
and $B_{vw}$ [$F_{vw}$] is the number of back [forward] arcs $vw$, then
\begin{align}
  \E B_{wv} = \E F_{vw} = \frac{\gl}{n} \E I_{vw}.
\end{align}
Summing over all pairs of distinct $v$ and $w$, we obtain
\begin{align}\label{EB=F}
  \E B = \E F = \frac{\gl}{n} \E\sum_{w} \sum_{v\neq w} I_{vw}
= \frac{\gl}{n} \E \sum_{w} d(w)
=\gl\E\bd.
\end{align}
Finally, $\E T+\E F = \E C + \E B$ by \eqref{tatf1} and \eqref{tabc1},
and thus \eqref{EB=F} implies $\E T = \E C$.
\end{proof}

\begin{remark}
Knuth \cite{Knuth12A} conjectures, based on exact calculation of generating
functions for small $n$,
that, much more strongly, 
$B$ and  $F$ have the same distribution for every $n$.
(Note that $T$ and $C$ do not have the same distribution; we have $T\le n-1$,
while $C$ may take arbitrarily large values.)
\end{remark}

\begin{remark}\label{RGeoL}
A simple argument with generating functions shows that the number of loops at
a given vertex $v$ is $\Ge(1-p/(n-np+p))$; these numbers are independent,
and thus
$L\sim\NegBin\bigpar{n,1-p/(n-np+p)}$ with $\E L = p/(1-p)=\gl=O(1)$ and
$\Var(L)=p(1-p+p/n)/(1-p)^2=\gl(1+\gl/n)=O(1)$
\cite{Knuth12A}. 
Moreover, it is easily seen that 
asymptotically, $L$ has a Poisson distribution,
$L\dto\Po(\gl)$ as \ntoo.
\end{remark}

\section{Depth, trees and arc analysis in the shifted 
geometric outdegree  distribution} 
\label{SGeo1}
In this section, the outdegree distribution is $\Gei(1-p)=1+\Ge(1-p)$.
Thus we now have the mean 
\begin{align}\label{glgeo1}
  \gl=\frac{1}{1-p}.
\end{align}
Thus $\gl>1$, and only the supercritical case occurs.
As in \refS{SGeo}, the depth $d(t)$ is a Markov chain given by \eqref{dt+},
but the distribution of $\xi_t$ is now different.
The probability \eqref{new} is replaced by $(1-t/n)/(1-pt/n)$, 
but the number of future arcs from an ancestor is still $\Ge(1-p)$, and, 
with $\gth:=t/n$,
\begin{align}
  \label{xi1}
\PP\bigpar{\xi_t=k}=
  \begin{cases}
    \bpi_t:=\frac{1-\gth}{1-p\gth}, & k=0,\\
(1-\bpi_t)(1-\pi_t)^{k-1}\pi_t,
& k\ge1,
  \end{cases}
\end{align}
where $\pi_t=p\bpi_t$ is as in \eqref{xipi}.
The rest of the analysis does not change,
and the results in \refTs{T1}--\ref{TGeoA} still hold,
but we get different values for many of the constants. 

We now have
\begin{align}\label{bmut}
\E \xi_t = \bmu_t:=\frac{(1-p)\gth}{p(1-\gth)}  
\end{align}
and 
instead of \eqref{Etd} we have
$\E \td(t) = n\tl(\gth) + O(1)$ where now
$\tl(\gth)$ takes the new value
\begin{align}\label{tl1}
    \tl(\gth)&:=
\int_0^\gth\lrpar{1-\frac{(1-p)x}{p(1-x)}  }\dd x
=\frac{1}{p}\theta+\frac{1-p}{p}\log(1-\theta)
\notag\\ &\phantom:
=\frac{1}{p}\bigpar{\theta+\gl\qw\log(1-\theta)}
.\end{align}
Note that $\tl(\gth)$ in \eqref{tl1} is proportional to \eqref{tl} for the 
(unshifted) geometric distribution with the same $\gl$, but larger by a
factor $1/p$.
Figures~\ref{fig:depth} and \ref{fig:depth2} show 
$\tl(\gth)$ for both geometric distributions
with the same $p$ ($0.6$) and the same $\gl$ (2.0), respectively.
\begin{figure}[ht]
\vspace{-1cm}
\includegraphics[height=8cm,trim = 1cm 18cm 0 0]{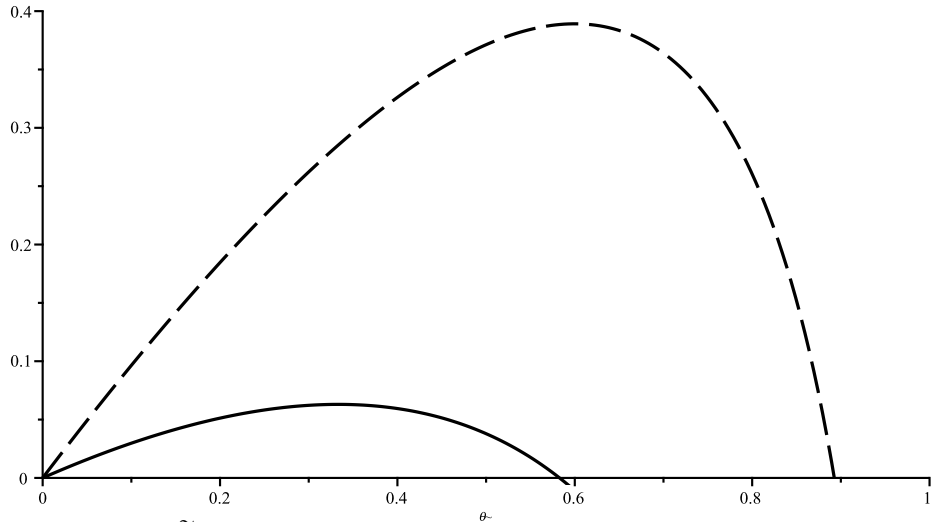}
\caption{$\tl(\gth)$, the asymptotic search depth, 
for geometric distribution $\Ge(1-p)$ (solid) 
and shifted geometric distribution $\Gei(1-p)$ (dashed) 
with $p=0.6$ and thus $\gl=1.5$ and $\gl=2.5$, respectively.}
\label{fig:depth}
\end{figure}

\begin{figure}[ht]
\includegraphics[height=8cm, width=12cm, trim= 3cm 95mm 0cm 0cm]{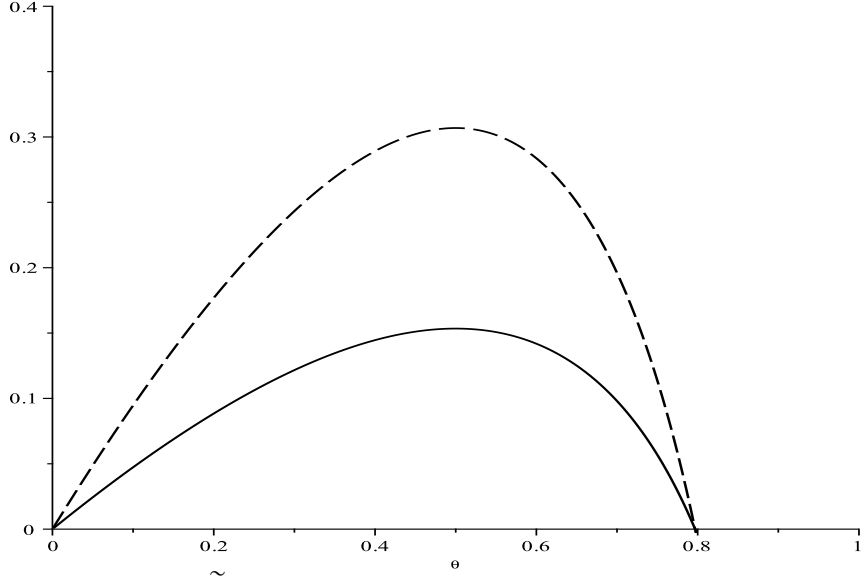}
\caption{$\tl(\gth)$, the asymptotic search depth, 
for geometric distribution $\Ge(1-p)$ (solid) 
and shifted geometric distribution $\Gei(1-p)$ (dashed) 
with $p=2/3$ and $1/2$, respectively, and thus $\gl=2$.}
\label{fig:depth2}
\end{figure}

Note that $\tl(\gth_1)=0$ still gives \eqref{gth1} and \eqref{gth1rhoo},
now with
$\gl=1/(1-p)$ as in \eqref{glgeo1}, and that $\gl>1$ so $\gth_1>0$
for every $p$. 
Differentiating \eqref{tl1} shows that the maximum point $\gth_0=p>0$, which
again is given by \eqref{gth0}.
Straightforward calculations yield
\begin{align}\label{gu1}
  \gu&:=\tl(p)= 1+\frac{1-p}{p}\log(1-p)
=1-\frac{1}{\gl-1}\log\gl,
\\\label{ga1}
\alpha&:=\frac{1}{p}\left(\frac{\theta_1^2}{2}
 -\frac{1}{\lambda}\bigpar{(1-\theta_1)\log(1-\theta_1)+\theta_1}\right)
=\gth_1-\frac{\theta_1^2}{2p}
.\end{align}
Furthermore, \eqref{xi1} yields by a simple calculation, with $\gth:=t/n$,
\begin{align}
  \Var \xi_t 
=\frac{(1-p)^2}{p^2(1-\gth)^2} + \frac{1-p}{p(1-\gth)} - \frac{1-p}{p^2}.
\end{align}
Hence, \eqref{v2} holds with \eqref{v3} replaced by
\begin{align}
  \gssx(\gth)&:=
\int_0^\gth\Bigpar{\frac{(1-p)^2}{p^2(1-x)^2} + \frac{1-p}{p(1-x)} -
  \frac{1-p}{p^2}}\dd x
\notag\\&\phantom:=
\frac{(1-p)^2\gth}{p^2(1-\gth)} -\frac{1-p}{p}\log(1-\gth) -\frac{1-p}{p^2}\gth,
\end{align}
and then  
\refL{LN1} and \refT{TN1} hold with this $\gssx(\gth)$.

Consequently, \refT{TH2} holds with
\begin{align}
  \gss:=\gssx(p)
= -\frac{1-p}{p}\log(1-p).
\end{align}

In the proof of \refT{TT}, \eqref{erika} for $k\ge1$
still holds with $\mu_t$ given by 
\eqref{mut} (except for the formula with $\gl$), and thus
\eqref{ele} is replaced by, using \eqref{bmut},
\begin{align}\label{ele1}
  \E\bigsqpar{d(t+1)\mid\cF_t}&
=d(t)+1-\bmu_t+\mu_t\E\bigsqpar{\JJ_{t+1}\mid\cF_t}
.\end{align}
The rest of the proof remains the same with minor modifications, and leads
to, instead of \eqref{cm7}, with $\gth:=t/n$,
\begin{align}\label{cm71}
  \sum_{t=T_1+1}^n \JJ_t&
=\sum_{t=T_1}^{n-1}\frac{\bmu_t-1}{\mu_t}+\Ollnqq
=\sum_{t=T_1}^{n-1}\frac{\frac{(1-p)\gth}{p(1-\gth)}-1}{\frac{1-p}{p(1-\gth)}}
+\Ollnqq
\notag\\&
=\sum_{t=T_1}^{n-1}\bigpar{1-\gl(1-\gth)} + \Ollnqq
,\end{align}
and thus
\refT{TT} holds with
\begin{align}\label{psi1}
  \psi:=\int_{\gth_1}^1\bigpar{1-\gl(1-x)}\dd x
=1-\gth_1-\frac{\gl}{2}(1-\gth_1)^2,
\end{align}
just as in \eqref{psi}. 

In the proof of \refT{TGeoA},
\eqref{babb} still holds, and we 
obtain \eqref{tgeoab} with $\gb=\gl\ga$, and then \eqref{tgeoac} with
$\chi=\gl/2-\gb$ just as before (but recall that $\ga$ now is different).
On the other hand, 
now the expected numbers of back and  forward arcs differ since 
$\E B = \gl\E\bd \sim\lambda\ga n$ and $\E F=(\gl-1)\E\bd\sim(\lambda-1)\ga n$ 
because the average number of future arcs at a vertex after a descendant
has been created is $\lambda-1$. 
The asymptotic formula \eqref{tgeoat} holds as above with $\tau:=1-\psi$;
hence \eqref{tatf1} implies that \eqref{tgeoaf} holds too, 
with $\gf=\gl/2-\tau$; as just noted, we now have 
$\gf=(\gl-1)\ga\neq\gb$. Collecting these constants, we see that
\refT{TGeoA} holds with
\begin{align}\label{gt1}
  \chit&:=1-\psi=\theta_1+\frac{\gl}{2}(1-\theta_1)^2,
\\
\gb&:=\gl\ga 
= \gl\gth_1-\frac{\gl}{2p}\gth_1^2
= \gl\gth_1-\frac{\gl^2}{2(\gl-1)}\gth_1^2
,\\
\gf&:=(\gl-1)\ga = (\gl-1)\gth_1-\frac{\gl}{2}\gth_1^2
=\frac{\gl}2-\chit, 
\\\label{chi1}
\chi&:=\frac{\gl}2-\gb
=\frac{\gl}2(1-\gth_1)^2 + \frac{\gl}{2(\gl-1)}\gth_1^2
.\end{align}
Thus the equality $\gb=\gf$ and the equality of the expected 
number of back and forward arcs in \refTs{TGeoA} and \refT{TT=C}
was an artefact of the  geometric degree distribution.
Similarly, $\chi=\gl/2-\gb<\gl/2-\gf=\chit$, and the equality
of the expected numbers of tree arcs and cross arcs in \refT{TT=C}
also does not hold.

We summarize the results above.
\begin{theorem}
  \label{TGeo1}
Let the outdegree distribution $\PQ$ be the shifted geometric distribution
$\Gei(p)$ with $p\in(0,1)$.
Then \refTs{T1}--\ref{TGeoA} hold, with the constants now having the values
described above (and always $\gl>1$), while \refT{TT=C} does not hold.
\end{theorem}

\section{A general  outdegree  distribution: stack size}
\label{Sgen}
In this section, we consider a general outdegree distribution $\PQ$, with mean
$\gl$ and finite variance. 

When the outdegree distribution is general, 
the depth does not longer follow a simple Markov chain, since
we would have to keep track of the number of children seen so far
at each level of the branch of the tree toward the current vertex.
Instead we get back a Markov chain if we instead of the depth 
consider the stack size $I(t)$, defined as follows.

The DFS can be regarded as keeping a stack of unexplored arcs, for which we
have seen the start vertex but not the endpoint.
The stack evolves as follows: 

{\addtolength{\leftmargini}{-10pt}
\begin{enumerate}
\renewcommand{\labelenumi}{\theenumi.}
\renewcommand{\theenumi}{\textup{S\arabic{enumi}}}%

\item \label{stack1}
If the stack is empty, pick a new vertex $v$ that has not been seen before
(if there is no such vertex, we have finished).
Otherwise, pop the last arc from the stack, and reveal its endpoint $v$
(which is uniformly random over all vertices).
If $v$ already is seen, repeat.
\item \label{stack2}
($v$ is now a new vertex)
Reveal the outdegree $\meta$ of $v$ and add to the stack $\meta$ new arcs
from $v$,
with unspecified endpoints. GOTO \ref{stack1}.
\end{enumerate}}

Let again $v_t$ be the $t$-th vertex seen by the DFS, and let 
$I(t)$ be the
size of the stack when $v_t$ is found (but before we add the arcs from
$v_t$).
It can easily be seen that $I(t)$ will be a Markov chain, similar (but not
identical) to the depth process $d(t)$ in the geometric case studied above.
Moreover, it is possible to recover the depth of the vertices from the stack
size process, which makes it possible to extend many of the results above,
although sometimes with less precision. 
For details, see \cite{SJ364-general}.

\end{document}